\documentclass[a4paper,11pt]{article}

\usepackage[utf8]{inputenc}
\usepackage[T1]{fontenc}
\usepackage[english]{babel}

\usepackage{microtype}

\usepackage[pdftex]{graphicx}
\usepackage[usenames,dvipsnames]{xcolor}

\usepackage{booktabs}

\usepackage{tikz}
\usetikzlibrary{calc}

\usepackage{amsmath,amssymb,amsfonts,amsthm}
\usepackage{mathabx}
\usepackage[ruled,linesnumbered,vlined]{algorithm2e}

\usepackage[pdftex]{hyperref}
\usepackage{cleveref}

\usepackage{fullpage}
\usepackage{enumitem}

\newcommand{\eps}{\varepsilon}

\newcommand{\ceil}[1]{\left\lceil{#1}\right\rceil}

\newcommand{\abs}[1]{\left | #1 \right |}

\newcommand{\set}[1]{\left \{ #1 \right \}}

\newtheorem{lemma}{Lemma}
\newtheorem{theorem}{Theorem}
\newtheorem{corollary}{Corollary}

\theoremstyle{definition}
\newtheorem{definition}{Definition}

\usepackage{authblk}
\title{New Subquadratic Approximation Algorithms for the Girth}

\author[1]{Søren Dahlgaard\thanks{Research partly supported by Mikkel Thorup's
Advanced Grant DFF-0602-02499B from the Danish Council for Independent Research
under the Sapere Aude research career programme.}}
\author[1]{Mathias Bæk Tejs Knudsen}
\author[1]{Morten Stöckel\thanks{Research partly supported by Villum Fonden.}}
\affil[1]{University of Copenhagen\\\texttt{[soerend,knudsen,most]@di.ku.dk}}
\date{}

\begin{document}

\maketitle
\begin{abstract}
    We consider the problem of approximating the girth, $g$, of an unweighted
    and undirected graph $G=(V,E)$ with $n$ nodes and $m$ edges. A seminal
    result of Itai and Rodeh~[SICOMP'78] gave an additive $1$-approximation in
    $O(n^2)$ time, and the main open question is thus how well we can do in
    subquadratic time.

    In this paper we present two main results. The first is a
    $(1+\eps,O(1))$-approximation in truly subquadratic time. Specifically, for
    any $k\ge 2$ our algorithm returns a cycle of length
    $2\ceil{g/2}+2\ceil{\frac{g}{2(k-1)}}$ in $\tilde{O}(n^{2-1/k})$ time. This
    generalizes the results of Lingas and Lundell [IPL'09] who showed it for
    the special case of $k=2$ and Roditty and Vassilevska Williams [SODA'12]
    who showed it for $k=3$. Our second result is to present an
    $O(1)$-approximation running in $O(n^{1+\eps})$ time for any $\eps > 0$.
    Prior to this work the fastest constant-factor approximation was the
    $\tilde{O}(n^{3/2})$ time $8/3$-approximation of Lingas and Lundell
    [IPL'09] using the algorithm corresponding to the special case $k=2$ of our
    first result.

\end{abstract}

\section{Introduction}
In this paper we consider the basic graph theoretical problem of computing the
shortest cycle of an unweighted and undirected graph. The length of this cycle,
$g$, is also known as the \emph{girth} of a graph.

Computing the girth of a graph has been studied since the 1970s. In a seminal
paper from 1978, Itai and Rodeh~\cite{ItaiR78} showed that the girth of an
$n$-node, $m$-edge graph can be computed in $O(n^\omega)$ time using fast
matrix multiplication, where $\omega < 2.373$ is the matrix multiplication
constant~\cite{Gall14a}. They also observed that running $n$ breadth first
searches gives an $O(mn)$ time \emph{combinatorial} algorithm for finding the
girth. Here \emph{combinatorial} means an algorithm that does not employ
Strassen-like cancellation tricks and arithmetic operations over some field.
Furthermore, it was shown by Vassilevska Williams and
Williams~\cite{WilliamsW10} that any combinatorial algorithm computing the
girth in $O(n^{3-\eps})$ for any $\eps > 0$ would imply a truly subcubic
algorithm (i.e. $O(n^{3-\eps})$ for some $\eps > 0$) for combinatorial boolean
matrix multiplication. Obtaining such an algorithm is widely conjectured to be
impossible.

This seeming barrier, combined with fast matrix multiplication being deemed
impractical, motivates the study of approximation algorithms for the girth,
$g$. In the paper of Itai and Rodeh~\cite{ItaiR78} they also presented an
algorithm computing a cycle of length at most $g+1$ in $O(n^2)$ time using a
simple BFS approach. In some sense this is an optimal approximation algorithm,
as the input may indeed be as large as $\Theta(n^2)$ and we cannot hope to get
a better approximation. However, we also know due to a classic result by
Bondy and Simonovits~\cite{BONDY197497} that any undirected graph with more
than $200n^{3/2}$ edges contains a $4$-cycle. Furthermore, we can find such a
cycle in $O(n)$ expected time using the subroutine of Yuster and
Zwick~\cite{YusterZ97} giving an $O(n)$ time algorithm returning a cycle of
length at most $g+1$ for graphs with many edges (since $g\ge 3$). It thus
remains interesting to obtain approximation algorithms with subquadratic
running time.

This problem was initially studied by Lingas and
Lundell~\cite{LingasL09} and later by Roditty and Vassilevska
Williams~\cite{RodittyW12}. In this paper we
say that an algorithm is a $(c_1,c_2)$-approximation to the girth for $c_1\ge
1, c_2\ge 0$, if the algorithm returns a cycle with length $\tilde{g}$ such
that $g\le \tilde{g}\le c_1\cdot g + c_2$. We say that a
$(c_1,0)$-approximation is a \emph{multiplicative} $c_1$-approximation (or just
$c_1$-approximation), and that a $(1,c_2)$-approximation is an \emph{additive}
$c_2$-approximation (or just $+c_2-approximation$).
Lingas and Lundell~\cite{LingasL09} initialized the study of subquadratic girth
approximation algorithms by giving a Monte Carlo $(2,2)$-approximation 
in expected time $O(n^{3/2}\sqrt{\log n})$. It is worth noting that this
algorithm represents an (at worst) $8/3$-approximation. They stated as an open
question whether a subquadratic time algorithm giving a multiplicative
approximation factor of two or better exists. This question was answered by
Roditty and Vassilevska Williams~\cite{RodittyW12} who presented several
subquadratic approximation algorithms for the girth. In particular they gave a
$2$-approximation in $O(n^{5/3}\log n)$ time. In fact, they
showed a more general result returning a cycle of length at most $2\ceil{g/2} +
2\ceil{g/4}$ (although they state it quite differently)\footnote{For a graph
with girth $g=4c-z$ for some $c\ge 1$ and $z\in\{0,1,2,3\}$ their algorithm
returns a cycle of length $6c-z$ for even $g$ and $6c-z+1$ for odd $g$. It can
be verified by inspection that this is indeed $2\ceil{g/2} + 2\ceil{g/4}$.} in
$O(n^{5/3}\log n)$.
This can be seen as an ``almost-but-not-quite'' $3/2$-approximation. They also
conjecture that obtaining a $(2-\eps)$-approximation requires
essentially quadratic time. Complementing this conjecture, they present a
randomized algorithm which beats this barrier for \emph{triangle-free graphs},
giving an $8/5$-approximation in $O(n^{1.968})$
which can be improved to $O(n^{1.942})$ using the result
of~\cite{dahlgaard2017evenCycles} as a lemma.
Finally, in the same paper, Roditty and Vassilevska Williams present an
additive $3$-approximation (additive $2$ for even $g$) in time $O(n^3/m\log^2
n)$.

\subsection{Our contribution}
In this paper we address the following two questions: ``How good an
approximation can we give for the girth in subquadratic time?'' and ``What is
the fastest running time of any constant-factor approximation algorithm to the
girth?''.

The conjecture of Roditty and Vassilevska Williams~\cite{RodittyW12} suggests
that we cannot hope to obtain a $(2-\eps)$-approximation faster
than $\Omega(n^2)$. However, we show that if we allow a small additive
error, we can an arbitrarily good multiplicative approximation in subquadratic
time -- that is, a $(1+\eps,O(1))$-approximation for any $\eps > 0$.
Specifically, we show the following theorem.
\begin{theorem}
	\label{thm:onePlusEps}
    Let $G$ be a graph with $n$ nodes and let $k\ge 2$ be any integer. Then
    there exists an algorithm that runs in time
	$O\!\left(n^{2-1/k}(\log n)^{1-1/k}\right)$ and finds
	a cycle with length at most $2\ceil{\frac{g}{2}} +
    2\ceil{\frac{g}{2(k-1)}}$, where $g$ is the girth of $G$.
\end{theorem}
Theorem~\ref{thm:onePlusEps} generalizes the result of Roditty and
Vassilevska Williams~\cite[Thm. 1.3]{RodittyW12}, who showed it for the special
case of $k=3$ (with an additional $\log^{1/3}n$ factor)
and the result of Lingas and Lundell~\cite[Thm. 1]{LingasL09} who showed a
Monte Carlo version for the special case of $k=2$.
As a corollary, we also get an improvement on the result for triangle-free
graphs \cite{RodittyW12,dahlgaard2017evenCycles} improving the running time
from $O(n^{1.942})$ to $O(n^{9/5}\log^{4/5} n)$ and
generalizing to families of graph with girth $g > 2\ell-1$ for any positive
integer $\ell$.
\begin{corollary}\label{cor:large_girth}
	Let $\ell$ be a positive integer and $k=2\ell-\ceil{\frac{\ell}{3}}+2$.
	There exists an algorithm that given a graph $G$ with $n$ nodes and girth $>2\ell-1$
	runs in time $O\!\left(n^{2-1/k}\right)$
	and gives a $(1+\frac{3}{2\ell+1})$-approximation of the girth.
\end{corollary}
For completeness we show how Theorem~\ref{thm:onePlusEps} directly implies
Corollary~\ref{cor:large_girth} above in Appendix~\ref{app:proofs}.

Complementing Theorem~\ref{thm:onePlusEps}, we also show that for any constant
$\eps > 0$ it is possible to obtain a constant-factor approximation in
$O(n^{1+\eps})$ expected time. Specifically, we show the following theorem.
\begin{theorem}\label{thm:constant}
	Let $G$ be a graph with $n$ nodes and $k$ an integer $\ge 2$. There exists
    an algorithm giving a $2^k$-approximation to the girth of $G$ with
    probability $1-\frac{1}{n}$ running in expected time $O\!\left(n^{1+1/k}k
    \log n\right)$.
\end{theorem}

Our algorithms are quite simple and rely on sampling and a ``stop-early'' BFS
procedure similar to previous work~\cite{ItaiR78,LingasL09,RodittyW12}. Using
techniques of~\cite{RodittyW12} it is possible to derandomize the sampling in
Theorem~\ref{thm:onePlusEps}, however this procedure is too slow and does not
work for Theorem~\ref{thm:constant}, hence this theorem is only available in a
randomized variant.

\subsection{Related work}
The problem of approximating the girth has also been considered in other
settings. For undirected graphs with weights in $\{1,\ldots, M\}$ Lingas and
Lundell~\cite{LingasL09} gave a $2$-approximation of the minimum
weight cycle running in $O(n^2\log n(\log n + \log M))$. Roditty and
Tov~\cite{RodittyT13} improved the approximation factor to $4/3$ while
maintaining the running time, and also gave a $O(\frac{1}{\eps}n^2\log
n\log\log n)$ time $(4/3+\eps)$-approximation for graphs with non-negative real
weights. Furthermore, it was shown by Roditty and Vassilevska
Williams~\cite{RodittyW11} that the problem (also for directed graphs) reduces
to finding a minimum weight triangle in an \emph{undirected} graph with $O(n)$
nodes and weights in $\{1,\ldots, O(M)\}$. For directed graphs a recent paper
by Pachocki, et al.~\cite{PachockiRSTW16} gave a $O(k\log
n)$-approximation in $O(mn^{1/k}\log^5 n)$ for any $k\ge 1$ and an additive
$O(n^\alpha)$-approximation in $O(mn^{1-\alpha})$ for any $\alpha\in (0,1)$.

Closely related to the problem of finding the girth of a graph is the problem
of finding a cycle of a fixed length $k$. For undirected graphs, Yuster and
Zwick~\cite{YusterZ97} showed that this can be done for even $k$ in
$O(f(k)\cdot n^2)$. Alon, Yuster and Zwick~\cite{AlonYZ97} showed that for
directed and undirected graphs this can be done in $O(f(k) \cdot m^{2-2/k})$ if
$k$ is even and $O(f(k) \cdot m^{2-2/(k+1)})$ if $k$ is odd. And for undirected
graphs when $k$ is even they give an algorithm with running time $O
(m^{2-(1+\ceil{k/4}^{-1})/(k/2+1)})$. For even $k$ this was improved
\cite{dahlgaard2017evenCycles} to $O(f(k) \cdot m^{2-4/(k+2)})$. This
problem of finding an even $k$-cycle was used as a subroutine by Roditty
and Vassilevska Williams~\cite{RodittyW12} in some of their algorithms for
approximating the girth.

\section{Preliminaries}
We will assume that all graphs $G=(V,E)$ in the paper are undirected,
unweighted, connected and contain at least one cycle. If $u$ is some node in
the graph $G$ and $\ell$ is an integer, we denote by $B_G(u,\ell)$ the
\emph{ball} of radius $\ell$ around $u$ in $G$, i.e.~the set of all nodes of
distance at most $\ell$ from $u$ in $G$. We will sometimes denote this simply
$B(u,\ell)$ when $G$ is clear from the context. We let $\Gamma(u)$ denote the
neighbourhood of $u$, i.e.~$\Gamma(u) = B(u,1) \setminus \set{u}$.
For a set of nodes $S\subseteq
V$, let $B(S,\ell)=\bigcup_{u\in S}B(u,\ell)$. We denote by $C_\ell$ the simple
cycle of length $\ell$ and let $\log x$ be the natural logarithm of $x$.

We will need the following lemma, which was used in a slightly weaker version
by Roditty and Vassilevska Williams~\cite{RodittyW12}
\begin{lemma}\label{lem:samples}
    Let $A_1,\ldots, A_k$ be sets over a universe $U$ of $n$ elements such that
    for every $i\le k$ we have $\abs{A_i}= x$ for some positive $x$. Then we
    can find a set $S\subseteq U$ with $|S|\le \frac{n}{x}\log k$ in $O(kx +
    n)$ time such that for each $1\le i\le k$ we have $S\cap A_i \ne
    \emptyset$.
\end{lemma}
\begin{proof}
    Let $t(u) = \abs{\{i\mid u\in A_i\}}$. We start by computing this quantity
    for each $u\in U$ in $O(kx+n)$ time by simply traversing all the sets
    $A_i$. We now keep $k$ linked lists $L_1,\ldots, L_k$, where $L_i$ contains
    the elements $u$ such that $t(u) = i$. In addition, we keep a table $T$
    where $T[u]$ contains a pointer to the node in $L_1,\ldots, L_k$ containing
    $u$.
    
    The algorithm now works as follows. Let $S = \emptyset$. Repeatedly find an
    element $u$ with maximum $t(u)$, add it to $S$, and for each remaining
    $A_i$ such that $u\in A_i$, we remove $A_i$ and for each $v\in A_i$ update
    $t(v)$ to $t(v)-1$ and move it from $L_{t(v)}$ to $L_{t(v)-1}$. We note
    that we may update an element $v$ several times, and that if $t(v)$ becomes
    $0$ we simply ignore $v$ for the remainder of the algorithm.

    Since the maximum value $t(u)$ can never decrease we can keep track of the
    maximum value in total $O(k)$ time over all iterations. Furthermore, all
    updates take $O(kx)$ time using the lists $L_1,\ldots, L_k$ and our table
    $T$. The analysis of the size of $S$ now follows exactly as
    in~\cite{RodittyW12}.
\end{proof}

Similar to previous papers~\cite{ItaiR78,RodittyW12,LingasL09} we will use a
procedure called $\FuncSty{BFS-Cycle}(G,u)$, which is simply the algorithm
which runs a BFS from $u$ in $G$ until a node $v$ is visited twice. In this
case the algorithm returns the simple cycle containing $v$ in the BFS tree
including the last edge visited. We will need the following well-known lemma
from the literature.
\begin{lemma}\label{lem:bfs_cyc}[\cite{ItaiR78,LingasL09}]
    Let $G = (V,E)$ and let $u\in V$ be any vertex. Then running $\FuncSty{BFS-Cycle}(G,u)$ takes
    $O(n)$ time. Furthermore, if $v$ is a vertex at distance $\ell$ from $u$
    and $v$ is contained in a simple cycle of length $k$, then
    $\FuncSty{BFS-Cycle}(G,u)$ returns a cycle of length at most
    $2\ceil{k/2} + 2\ell$.
\end{lemma}
Sometimes we will need to run a restricted version of $\FuncSty{BFS-Cycle}$
which stops after visiting a certain number of nodes (or if a cycle was found
before). Let this procedure be denoted by $\FuncSty{BFS-Cycle}(G,u,y)$, where
$y$ is the bound on the number of visited nodes. It is clear that
$\FuncSty{BFS-Cycle}(G,u,y)$ takes at most $O(y)$ time if we assume an
adjacency list representation.

We will need the following algorithm, which samples a set $S$ and runs
$\FuncSty{BFS-Cycle}$ from each $u\in S$.
\begin{definition}\label{defn:algo}
    Let $A(G,x,y)$ be an algorithm that takes an $n$-node graph $G=(V,E)$ and
    two positive numbers $x,y$. The algorithm creates a set $S\subseteq V$ of
    nodes by sampling each node of $V$ independently with probability
    $\min\left\{\frac{x}{n},1\right\}$. Clearly $S$ has expected size $x$. The
    algorithm then runs $\FuncSty{BFS-Cycle}(G,u,y)$ for each $u\in S$ and
    returns the smallest cycle found or nothing if no cycle is found.
\end{definition}
We will also need the following deterministic variant, where the set $S$ is
picked using Lemma~\ref{lem:samples}.
\begin{definition}\label{defn:algo_det}
    Let $A_{det}(G,x,y)$ be the same algorithm as $A(G,x,y)$ with the following
    modification to how the set $S$ is picked. Fix $\tilde{x} = \frac{n\log
    n}{x}$ and denote the nodes of $G$ by $u_1,\ldots,u_n$. The algorithm first
    creates the sets $A_1,\ldots, A_n$, where $A_i$ is the set containing the
    $\tilde{x}$ closest nodes to $u_i$ (breaking ties arbitrarily, eg. by the
    order in the adjacency lists). It then creates the set $S\subseteq V$ by
    running the algorithm of Lemma~\ref{lem:samples} on $A_1,\ldots, A_n$.
\end{definition}
It is easy to see that $A(G,x,y)$ runs in $O(n + xy)$ expected time and that
$A_{det}(G,x,y)$ runs in $O\!\left(\frac{n^2\log n}{x} + xy\right)$ time.

\section{Subquadratic girth approximations}
\label{sec:subquad}
In this section we present our algorithm for obtaining subquadratic
approximations of the girth. In particular we will prove
Theorem~\ref{thm:onePlusEps}.

\begin{proof}[Proof of Theorem~\ref{thm:onePlusEps}]
    Let $x$ be some parameter to be defined later. We run two algorithms and
    take the smallest cycle produced by the two algorithms. The first algorithm
    is simply $A_{det}\!\left(G,\frac{n\log n}{x},n\right)$ of
    Definition~\ref{defn:algo}. The second algorithm is the following. For each
    node $u\in V$ we calculate the largest integer $r(u)$ such that
    $\abs{B(u,r(u))} \le x$. Let $u_1,\ldots, u_n$ denote the nodes of $G$ in
    sorted order such that $r(u_1)\ge\ldots\ge r(u_n)$.
    We initialize an empty graph $H$. For each $i = 1,2,\ldots,n$ in order such
    that $r(u_i) > 0$ we add $u_i$ to $H$ along with edges
    $\set{u_1,\ldots,u_{i-1}} \cap \Gamma(u_i)$ and run
    $\FuncSty{BFS-Cycle}(H,u_i,\frac{n\log n}{x})$ keeping track of the
    smallest cycle seen this far. The output of our algorithm is the smallest
    cycle produced by any of the above steps.

    It remains to prove that the algorithm returns a cycle of length at most
    $2\ceil{g/2} + 2\ceil{\frac{g}{2(k-1)}}$. Let $C$ be a cycle of
    length $g$ in $G$ and let $t = \min_{u \in C} \set{r(u)}$. We know that the
    set $S$ produced by
    $A_{det}(G,\frac{n\log n}{x},n)$ contains 
    a node with distance at most $t+1$ to $C$ and therefore, by
    Lemma~\ref{lem:bfs_cyc}, this algorithm
    returns a cycle with length at most $2\ceil{g/2}+2(t+1)$.
    It follows that the call to $A_{det}(G,\frac{n\log n}{x}, n)$ suffices for
    $t < \ceil{\frac{g}{2(k-1)}}$. Now, assume that $t \ge t_0$, where $t_0
    = \ceil{\frac{g}{2(k-1)}}$ and let $u$ be the last node of $C$ that was
    added to $H$. 
    Let $H_u$ denote the graph $H$ after adding $u$. It follows by the
    definition of $u$ that the cycle $C$ is contained in $H_u$. Furthermore,
    for each $v\in H_u$ we have $r(v)\ge r(u)\ge t_0$ and thus
    $\abs{B_{H_u}(v,t_0)} \le x$. This implies that
    $\abs{B_{H_u}(v,(k-1)t_0)} \le x^{k-1}$ for each $v \in H_u$
    and since $(k-1)t_0 \ge \ceil{g/2}$ by assumption, this implies that
    $\abs{B_{H_u}(u,\ceil{g/2})} \le x^{k-1}$. We can now use
    Lemma~\ref{lem:bfs_cyc} to conclude that the call to
    $\FuncSty{BFS-Cycle}(H_u,u,x^{k-1})$ returns a cycle of length at
    most $2\ceil{g/2}$.
    
    For the running time of the algorithm observe first that we can find $r(u)$
    for each $u$ in $O(x)$ time using a BFS. Furthermore, since we only
    add nodes $u_i$ to $H$ when $r(u_i)>0$ it follows that we only consider at
    most $O(nx)$ edges for addition to $H$. The running time is now
    bounded by $O\!\left(n\cdot \frac{n\log n}{x} + nx^{k-1}\right)$,
    which is minimized when $x = (n\log n)^{1/k}$ giving a total running time
    of $O(n^{2-1/k}(\log n)^{1-1/k})$.
\end{proof}

\section{Constant approximation in near-linear time}\label{sec:nearlinear}

%

We will now show how to obtain a constant factor approximation randomized in
near-linear time.

\begin{proof}[Proof of Theorem~\ref{thm:constant}]
    The algorithm is very straight-forward: For each $i=1,\ldots,k$ we run
    $A(G,n^{1+(1-i)/k}\log n, n^{i/k})$ and return the minimal cycle found or
    nothing if no cycle is reported.

    We now analyze this procedure. Let $g$ be the girth of $G$ and $C$ a cycle
    in $G$ with length $g$. Let $r$ be the smallest non-negative integer such
    that
    \[
        \abs{B\left(C,\ceil{g/2} \cdot (2^r-1)\right)} \le n^{r/k}\ .
    \]
    Clearly such an $r$ exists, as we may pick $r=k$. Also observe that $r > 0$
    since, in particular, the ball with $r=0$ contains the cycle $C$
    itself. We will show that the
    $r$th iteration of the algorithm gives a sufficiently small cycle.
    Consider the algorithm $A(G,n^{1+(1-r)/k}\log n, n^{r/k})$ and let $S$
    be the set of nodes sampled by this algorithm. We will show
    that there exists a node $u\in S$ close to $C$, such that
    $C$ is contained in the tree explored by
    $\FuncSty{BFS-Cycle}(G,u,n^{r/k})$.
    Consider the slightly smaller ball $B\left(C,\ceil{g/2} \cdot
    (2^{r-1}-1)\right)$. By the minimality of $r$, the number of nodes in
    $S$ belonging to this ball in expectation is at least
    \begin{align*}
        \abs{B\left(C,\ceil{g/2} \cdot (2^{r-1}-1)\right)} \cdot n^{(1-r)/k}\log n
        \ge 
        n^{(r-1)/k}\cdot n^{(1-r)/k}\log n
        =
        \log n
        \ .
    \end{align*}
    Therefore the probability that no node from $B\left(C,\ceil{g/2} \cdot
    (2^{r-1}-1)\right)$ is sampled is at most $\frac{1}{n}$. We thus assume
    that such a sampled node, $u$, exists. We now argue that the BFS search
    starting in $u$ gives the desired cycle. This is illustrated in
    Figure~\ref{fig:constant}.
    \begin{figure}[htbp]
        \centering
        \includegraphics[width=.5\textwidth]{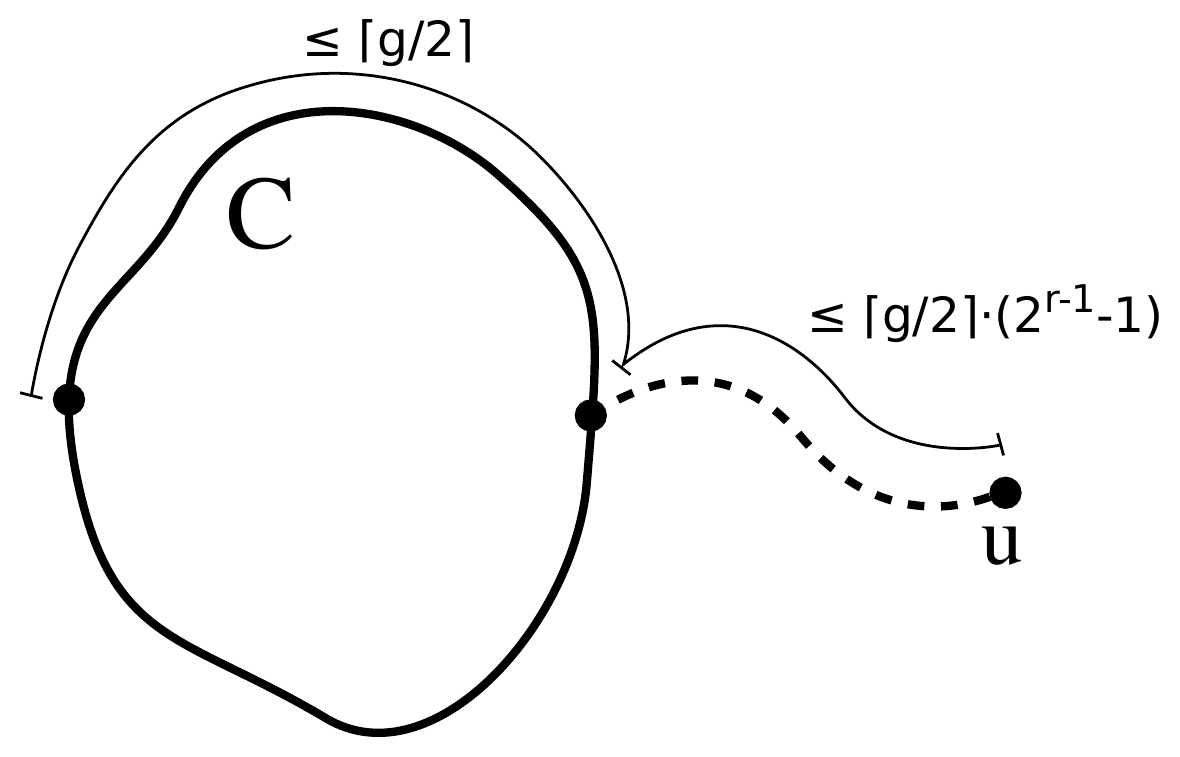}
        \caption{Illustration of the proof of Theorem~\ref{thm:constant}. $u$
        is a sampled node, and running $\FuncSty{BFS-Cycle}$ from $u$ visiting
        at most $n^{r/k}$ gives the desired approximation.}
        \label{fig:constant}
    \end{figure}
    Since $u$ is contained in $B\left(C,\ceil{g/2}\cdot(2^{r-1}-1)\right)$ we
    know that $B(u,\ceil{g/2}2^{r-1})\subseteq
    B\left(C,\ceil{g/2}\cdot(2^r-1)\right)$ and thus, by the definition of $r$,
    we know that $\abs{B(u,\ceil{g/2}2^{r-1})} \le n^{r/k}$.
    Furthermore, this ball around $u$ also contains the cycle $C$ and thus
    $A(G,n^{1+(1-r)/k}\log n, n^{r/k})$ 
    returns a cycle of length at most $2 \cdot \ceil{g/2} \cdot 2^{r-1}$. Since
    $r \le k$ the length of the cycle returned is at most $2^k\ceil{g/2}$.

    Since we invoke that algorithm $A$ exactly $k$ times and each invocation
    takes $O(n^{1+1/k}\log n)$ in expectation the running time follows.
\end{proof}

We note that we cannot employ the algorithm $A_{det}$ instead in the algorithm
above, as the task of creating a ball around each node $u\in V$ in order to
employ Lemma~\ref{lem:samples} takes too long. This is the main bottle-neck in
obtaining a deterministic variant of Theorem~\ref{thm:constant}.

\section{Conclusion and open problems}
In this paper we have studied the problem of obtaining subquadratic
approximation algorithms for the girth of an undirected and unweighted graph.
We have shown how to obtain a multiplicative $(1+\eps)$-approximation with
small additive error in subquadratic time, and an $O(1)$-approximation in
$O(n^{1+\eps})$ time for any $\eps > 0$. It remains as the main open question
whether one can obtain a multiplicative $(2-\eps)$-approximation in
$O(n^{2-\eps})$ time or perhaps show that obtaining such an approximation
requires $n^{2-o(1)}$ using the framework of Hardness in P. Another interesting
question is whether one can improve on our Theorem~\ref{thm:constant} and
obtain a multiplicative $O(1)$-approximation in $n^{1+o(1)}$ time or an
additive $O(1)$-approximation in $O(n^{2-\eps})$ time. Finally, it is an
interesting question whether one can improve on Corollary~\ref{cor:large_girth}
and obtain a multiplicative $(1+\frac{3}{2\ell+1}-\eps)$-approximation in
subquadratic time for graphs with girth $> 2\ell - 1$.

%

\bibliography{girth}

\newpage
\appendix
\centerline{\Huge Appendix}
\section{Omitted proofs}\label{app:proofs}

\begin{proof}[Proof of Corollary~\ref{cor:large_girth}]
We use Theorem \ref{thm:onePlusEps} with $k=2\ell-\ceil{\frac{\ell}{3}}+2$.
In order to prove that we get a $\left(1+\frac{3}{2\ell+1}\right)$-approximation
we just need to prove that
\begin{align}
	\label{eq:approxRatio}
	\frac{2\ceil{\frac{g}{2}} + 2\ceil{\frac{g}{2(k-1)}}}{g} \le 
	1 + \frac{3}{2\ell+1}
	\, ,
\end{align}
whenever $g \ge 2\ell$. For $g = 2\ell$ \eqref{eq:approxRatio} holds, so assume
that $g \ge 2\ell+1$.

For $2\ell+1 \le g \le 2(k-1)$ we note that $2\ceil{\frac{g}{2}} \le g+1$ and
$2\ceil{\frac{g}{2(k-1)}} = 2$. Therefore:
\begin{align*}
	\frac{2\ceil{\frac{g}{2}} + 2\ceil{\frac{g}{2(k-1)}}}{g} \le 
	\frac{g+3}{g} = 
	1 + \frac{3}{g} \le 
	1 + \frac{3}{2\ell+1}
	\, ,
\end{align*}
and \eqref{eq:approxRatio} holds. So assume that $g > 2(k-1)$.
Then we can write $g$ as $g = 2(k-1)q + r$ for some positive integer $q$ and some
$r \in \set{1,2,\ldots,2(k-1)}$. We again use that $2\ceil{\frac{g}{2}} \le g+1$
and see that $\ceil{\frac{g}{2(k-1)}} = q+1$. So we get that:
\begin{align}
	\label{eq:approxRatioGLarge}
	\frac{2\ceil{\frac{g}{2}} + 2\ceil{\frac{g}{2(k-1)}}}{g} \le 
	\frac{g+1+2q+2}{g} =
	1 + \frac{3+2q}{g} \le 
	1 + \frac{3+2q}{2(k-1)q+1}
	\, .
\end{align}
The right hand side of \eqref{eq:approxRatioGLarge} is maximized when $q=1$, and therefore
we get:
\begin{align*}
	\frac{2\ceil{\frac{g}{2}} + 2\ceil{\frac{g}{2(k-1)}}}{g} \le 
	1 + \frac{5}{2k-1}
	\, ,
\end{align*}
and therefore we just need to prove that $\frac{5}{2k-1} \le \frac{3}{2\ell+1}$. This
can be rewritten as
\begin{align*}
	k \ge \frac{5\ell+4}{3}
	\, .
\end{align*}
But we have that $\ceil{\frac{\ell}{3}} \le \frac{\ell+2}{3}$ and therefore
\begin{align*}
	k =
	2\ell - \ceil{\frac{\ell}{3}} + 2 \ge 
	2\ell - \frac{\ell+2}{3} + 2 = 
	\frac{5\ell+4}{3}
	\, ,
\end{align*}
as desired.
\end{proof}

\end{document}